\documentclass[invmat]{svjour}
\usepackage{color}
\usepackage{graphics}
\usepackage{amsmath}
\usepackage{amssymb}
\usepackage{graphicx}
\usepackage{psfrag}

\def\bfR{\mbox{\boldmath$R$}}

\def\eps{\epsilon}

\def\div{{\rm div}}

\def\ds{\displaystyle}

\def\bfR{\mbox{\boldmath$R$}}

\def\div{{\,\textrm{div}\,}}

\begin{document}
\title{Bacterial chemotaxis without gradient-sensing}

\author{Changwook Yoon\thanks{This work has been done as a part of the first author's Ph.D. thesis.} and Yong-Jung Kim}
\institute{$^1$Department of Mathematical Sciences, KAIST, ~291 Daehak-ro, Yuseong-gu, Daejeon 305-701, Korea\\
\\
E-mails: chwyoon@gmail.com (C. Yoon),
yongkim@kaist.edu (Y.-J. Kim)\\
}
\date{July 23, 2013}
\maketitle

\markboth{C. Yoon and Y.-J. Kim}{Bacterial chemotaxis without gradient-sensing}
\begin{abstract}
Models for chemotaxis are based on gradient sensing of individual organisms. The key contribution of Keller and Segel \cite{KS1,KS2} is showing that erratic movements of individuals may result in an accurate chemotaxis phenomenon as a group. In this paper we provide another option to understand chemotactic behavior when individuals do not sense the gradient of chemical concentration by any means. We show that, if individuals increase their motility to find food when they are hungry, an accurate chemotactic behavior is obtained without sensing the gradient. Such a random dispersal has been suggested by Cho and Kim \cite{CK} and is called starvation driven diffusion. This model is surprisingly similar to the original derivation of Keller-Segel model. A comprehensive picture of traveling band and front solutions is provided with numerical simulations.
\end{abstract}


\section{Introduction}

The chemotactic behavior of single-celled organisms such as \emph{Escherichia coli} \cite{A1,A2} and \emph{Dictyostelium} \cite{Bonner} has been well known for a long time. This bacterial chemotaxis is one of the most studied systems in biology and one may find a detailed mechanism of the system and related mathematical approaches from Tindall \emph{et al.} \cite{MR2430318,MR2430317}. It is quite a surprising and mysterious phenomenon that microscopic scale organisms figure out the macroscopic scale concentration gradient and find the correct direction to find food. The main contribution of the classical Keller-Segel model on this issue is to show that, even if the individuals sense the gradient inaccurately, the chemotactic phenomenon is accurate in an averaged sense.

The purpose of this paper is to go one step further from Keller and Segel's idea and conclude that, even if individual organisms cannot sense the gradient at all, accurate chemotactic behavior is obtained. The required hypothesis is that individuals increase their motility to find food when they are hungry. More specifically, we will classify the types of traveling wave solutions and find the necessary and sufficient conditions for the existence of one-dimensional traveling wave solutions of a chemotaxis model,
\begin{equation}\label{Eqn}
	\left\{
        \begin{aligned}
		u_t&= ~~(\gamma\,u)_{xx},\\
		m_t&= -k(m)u,
		\end{aligned}
	\right.
\quad x\in\bfR,\ t>0.
\end{equation}
In this model $u\ge0$ is the population density of a biological species, $m\ge0$ is the nutrient concentration and $k(m)\ge0$ is the consumption rate. The main feature of the model is the nonconstant motility $\gamma>0$ of the species. It is assumed that the species reduce the motility if there are more food to stay in the place. However, if the place is too crowed, then they increase the motility since the place is too competitive. Hence we set $\gamma=\gamma(m,u)$ with assumptions of
\begin{equation}\label{gamma_u}
\gamma_u\ge0\quad\mbox{and}\quad\gamma_m\le0.
\end{equation}
These hypotheses imply that starvation increases the motility and hence the corresponding dispersal process is called a \emph{starvation driven diffusion}. Then, the first equation in (\ref{Eqn}) is written as
\begin{equation}\label{Eqn2}
u_t=\big((\gamma+\gamma_uu)u_x+\gamma_mu\,m_x\big)_x,
\end{equation}
which now contains both advection and diffusion. In the appendix section a brief modeling procedure of a starvation driven diffusion is given.

Consider the traveling wave phenomenon of chemotactic bacteria found by Adler \cite{A1}. In his experiments E-coli colony is placed at one end of capillary tube filled with galactose and produces a traveling wave band of the bacteria moving to the other end. This phenomenon has been mathematically modeled by Keller and Segel \cite{KS1,KS2}. Let us compare the above chemotaxis model to the the classical Keller-Segel:
\begin{equation}\label{KS}
	\left\{
        \begin{aligned}
		u_t&= (\mu (m) u_x-\chi(m) u\,m_x)_x,\\
		m_t&= \eps\, m_{xx} - k(m)u,
		\end{aligned}
	\right.	
\end{equation}
where $\eps>0$ is the diffusivity of the nutrient. (The population reaction term is dropped in this discussion.) The diffusivity $\mu(m)>0$ and the chemosensitivity $\chi(m)>0$ satisfy a relation
\begin{equation}\label{chi-mu}
\chi(m)=-(1-a)\mu'(m),
\end{equation}
where $a<1$ is the constant ratio of effective body length of bacteria and the diffusivity is a decreasing function of nutrient concentration
$$
\mu'(m)\le0.
$$
Therefore, except the constant coefficient $1-a$ in (\ref{chi-mu}), the first equation in (\ref{KS}) is identical to (\ref{gamma_u}) for the special case that the motility is a function of $m$ only, i.e., when $\gamma=\gamma(m)$.

Traveling wave solutions of (\ref{KS}) have been intensively studied after various simplifications. First, Keller and Segel \cite{KS2} found explicit traveling waves by themselves for the case that
$$
\eps=\mu'=k'=0\quad\mbox{and}\quad \chi(m)=m^{-1}.
$$
The last chemosensitivity is based on Weber-Fechner law that individuals senses the relative gradient, but not the absolute gradient. This term gives a super sensitivity that diverges as $m\to0$, which makes the traveling wave permanent against the diffusive dynamics with constant diffusivity $\mu>0$. The traveling wave solution has bee studied for non-constant consumption rates,
$$
k(m)=m^p,\qquad p\ge0.
$$
The case with $0\le p<1$ gives traveling band solutions and the existence, stability and convergence as $\eps\to0$ has been intensively studied by many authors (see \cite{MR1167896,MR0411681,LS,MR2684162,MR1138847,R1,R2,R3,R4,MR2817814}).  The case with $p=1$ gives traveling fronts instead of bands. In particular, the case models vascular growth of angiogenesis and has been well-documented (see \cite{Chaplain,MR1969568,MR2099126,MR1828815,MR3000369}).

The diffusivity $\mu$ and chemotactic coefficient $\chi$ are considered independently in most chemotaxis models and the original relation of (\ref{chi-mu}) has been forgotten in the literature. Under this relation, the diffusivity becomes $\mu\cong{1\over a-1}\ln(m)+c_0$ if $\chi(m)={1\over m}$. Therefore, the diffusion also blows up as $m\to0$ and we will answer if there exists a traveling wave solution even with such a super sensitivity of diffusion. The analysis of this paper is based on such a relation between the diffusivity and the chemotactic coefficient in the frame of starvation driven diffusion written in (\ref{Eqn2}), which even includes the population dependency in the diffusion and chemosensitivity.

We assume that organisms only know the amount of food they obtains. Let
\begin{equation}\label{satisfaction}
s={m\over u},
\end{equation}
which measures the amount of food available for each individual and is called a \emph{satisfaction measure}. Individuals may not know the size resources $m$ nor the size population. However, they will know if they are getting enough food or not. We will abuse notation by writing $\gamma(m,u)=\gamma({m\over u})$ and the motility is a decreasing function on the satisfaction measure $s$. Then,
$$
\gamma_m=\gamma'(s){1\over u}\le0\mbox{~~~and~~~}
\gamma_u=-\gamma'(s){m\over u^2}\ge0,
$$
which satisfies the hypotheses in (\ref{gamma_u}). Furthermore, $\gamma$ is assumed to be smooth on $(0,\infty)$ and bounded below away from zero, i.e.,
\begin{equation}\label{alpha}
\lim_{s\to\infty}\gamma(s)=\ell>0,\quad \gamma'(s)\le0\mbox{~~on~~}(0,\infty).
\end{equation}

The consumption rate $k(m)$ is assumed to satisfy two hypotheses:
\begin{equation}\label{k}
k'(m)\ge0,\quad \lim_{m\to0}k(m)=0.
\end{equation}
The first one assumes that organisms consume more if there are more food. The second one indicates that the organisms do not consume at all if there is nothing to eat of course. Two typical examples in the literature are $k(m)=m^p$ and $k(m)={m^p\over 1+m^p}$ with $p>0$. The last one has bounded consumption rate as $m\to\infty$, which is more realistic. However, since we always consider uniformly bounded nutrient concentration $m\ge0$ and hence there is no essential difference in our context. However, the choice of $p$ makes a difference. The regime with $p\le1$ is more realistic since the consumption rate usually changes slower than the the amount of nutrient. However, the case with $p>1$ provides mathematically interesting phenomena and hence we also consider both cases.

\section{Main Results}

The traveling wave solution of a traveling speed $c\in\bfR$ is a solution in the form of
$$
u(x,t)=u(x-ct),\quad m(x,t)=m(x-ct),
$$
where we are abusing notations. Introduce the variable for the traveling wave solution
$$
\xi=x-ct.
$$
Then, from the equation (\ref{Eqn}), we obtain a system of ODEs
\begin{eqnarray}\label{TravelingEqn}
&&c\,u'= -(\gamma(s)u)'',\quad \xi\in\bfR,\\
&&c\,m'= k(m)u,\qquad~ \xi\in\bfR,\label{TravelingEqn1}
\end{eqnarray}
where the notation $'$ is for the ordinary differentiation for single variable functions. The boundary conditions are
\begin{equation}\label{BC}
u'\to0,\ m'\to 0,\ u\to u_\pm\ge0,\ m\to m_\pm\ge0 ~~\text{as } \xi \to \pm\infty.
\end{equation}

\begin{lemma}\label{lem1}
Let $(u(\xi),m(\xi))$ be a nonnegative traveling wave solution of (\ref{TravelingEqn})--(\ref{BC}). If $u\not\equiv0$, then $m(\xi)$ is strictly monotone.
\end{lemma}
\begin{proof}
Since the motility function satisfies (\ref{alpha}), the equation (\ref{Eqn2}) is uniformly parabolic. Since $u\not\equiv0$, $u(\xi)>0$ for all $\xi\in\bfR$ by the uniform parabolicity. Therefore, by (\ref{TravelingEqn1}), we have
$$
m'(\xi)={k(m) u\over c}\ne0.
$$
Hence, $m$ is strictly monotone.
$\hfill\qed$\end{proof}

If $u_-=u_+$, the solution is called a pulse type or a localized traveling wave. Otherwise, the traveling wave is called a front type. In other words, Lemma \ref{lem1} implies that traveling wave of $m$ is always a front type except a trivial one. Hence, without loss of generality, we may assume
\begin{equation}\label{Orientation}
0\le m_-<m_+.
\end{equation}
Then, the solution $m$ is strictly increasing and hence we may consider $\xi=\xi(m)$ as the inverse function of $m=m(\xi)$. This choice of boundary data, $m_+>0$, forces us to choose
\begin{equation}\label{u+}
u_+=0
\end{equation}
since that is the only case that makes $(m_+,u_+)$ be a steady state. One can easily see that this steady state is unstable.

\begin{lemma} \label{lem2}If $m_->0$ or $\int_0^{m_+}{1\over k(\eta)}d\eta<\infty$, then the population size of the species is finite. Hence $u_-=0$ is a necessary condition for the existence of a solution.
\end{lemma}
\begin{proof} The equation (\ref{TravelingEqn1}) can be written as $u=c{m'\over k(m)}$. Its integration over the whole real line gives the total population of the species, which is
$$
\int_{-\infty}^\infty u(\xi)d\xi=c\int_{-\infty}^\infty {m'\over k(m)}d\xi =c\int_{m_-}^{m_+}{1\over k(m)}dm.
$$
Under the condition (\ref{k}) the right side is finite for all $m_->0$. Therefore, the population size is finite under the assumptions. If $u_->0$, then the population size becomes infinite and hence $u_-=0$ is necessary condition. $\hfill\qed$
\end{proof}

Lemma \ref{lem2} implies that there exists only a pulse type traveling wave solution if $\int_0^{m_+}{1\over k(\eta)}d\eta<\infty$. The main theorems of this paper are about sufficient and necessary conditions for the existence of a traveling wave solution of (\ref{TravelingEqn})--(\ref{TravelingEqn1}) with a given speed $c\in\bfR$. We first consider necessary conditions.

\begin{theorem}\label{thm1} Suppose that $0\le m_-<m_+<\infty$ and $u_+=0$ and consider the existence of a solution of (\ref{TravelingEqn})--(\ref{BC}).

\noindent $(i)$ $m_-=0$ and $c>0$ are necessary conditions.

\noindent $(ii)$ If $\gamma(s)$ is bounded, there is no traveling wave solution.

\end{theorem}

\begin{theorem}\label{thm2} Suppose that $0\le m_-<m_+<\infty$, $u_+=0$, and $\int_0^{m_+}{1\over k(\eta)}d\eta<\infty$. Then, there exists a traveling wave solution of (\ref{TravelingEqn})--(\ref{BC}) if and only if $u_-=m_-=0$ and $\gamma(s)\to\infty$ as $s\to0$. The total population of the traveling species is given by
\begin{equation}\label{population}
N:=c\int_0^{m_+}{1\over k(\eta)}d\eta.
\end{equation}
\end{theorem}

Theorem \ref{thm2} gives the correct boundary condition and the motility function for the existence of a traveling wave solution. The relation (\ref{population}) shows that the wave speed $c$ is related to the population size $N>0$ and the consumption rate $k(m)$, but not to the motility function $\gamma$.

\begin{remark}
Keller and Odell \cite[Theorem]{MR0411681} found necessary and sufficient conditions for the existence of traveling wave band of finite size, which are
\begin{eqnarray}\label{kChi}
&&\lim_{m\to0+}{1\over k(m)\chi(m)}=0,\\
&&\lim_{m\to0+}\int_0^{m_+}{\exp\big(-\int_m^{m_+} {\chi(\eta)\over\mu(\eta)}d\eta\big)\over \mu(m) k(m)}\,dm=0.
\end{eqnarray}
For the case with
$$
\mu= m^r,\quad \chi= m^{-q}\quad \mbox{and}\quad k=m^p,
$$
the corresponding cases are
$$
p<\min(1,q)\quad\mbox{and}\quad r+q\ge1.
$$
These conditions are simplified in our case as the unboundedness of $\gamma$ as $s\to0$. $\hfill\qed$
\end{remark}

Next, we will go one step further for the case $\int_0^{m_+}{1\over k(m)}dm=\infty$. The condition (\ref{kChi}) appears in a slightly modified form in the following theorem. In this case we obtain pulse type traveling wave solutions of infinite mass and front type ones. Non-existence of a traveling wave solution is also shown.

\begin{theorem}\label{thm3} Suppose that $0\le m_-<m_+<\infty$, $u_+=0$, $\int_0^{m_+}{1\over k(\eta)}d\eta=\infty$, and $\gamma(s)$ is convex in a neighborhood of $s=0$.

\noindent $(i)$ Suppose that there exists $0<u_0<\infty$ such that
\begin{equation}\label{speed}
-\lim_{m\to0}\gamma'\Big({m\over u_0}\Big)k(m)=c^2.
\end{equation}
Then, there exists a traveling wave solution of (\ref{TravelingEqn})--(\ref{BC}) if and only if  $u_-=u_0$ and $m_-=0$.

\noindent $(ii)$ Suppose that
\begin{equation}\label{speed2}
-\lim_{m\to0}\gamma'\Big({m\over u_0}\Big)k(m)>c^2
\end{equation}
for any $u_0>0$. Then, there exists a traveling wave solution of (\ref{TravelingEqn})--(\ref{BC}) if and only if  $u_-=m_-=0$. Furthermore, this pulse type traveling wave solution has infinite total population.

\noindent $(iii)$ Otherwise, there is no traveling wave solution of any type.
\end{theorem}

The relation (\ref{speed}) shows that the speed of the traveling wave is independent of the amount of nutrient $m_+>0$. However, the amount of population density at the left side, $u_-=u_0$, decides the speed $c>0$. If (\ref{speed2}) is satisfied for all $u_0>0$ for a given $\tilde c>0$, then it is so for any given positive speed $c<\tilde c$. In many cases, it is also true for any $c>0$. Therefore, Theorem \ref{thm3}$(ii)$ implies that there exists a localized traveling wave solution with any speed $c>0$ with infinite total population. Such a localized traveling wave solution with infinite total mass is reported here for the first time.

\section{Proof of Theorems}

Let $0\le m_-<m_+<\infty$ and $u_+=0$. Then, the integration of (\ref{TravelingEqn}) on $(\xi,\infty)$ gives
$$
-cu=(\gamma(s)u)'=\gamma_mm'u+\gamma_uuu'+\gamma(s)u',
$$
where the boundary terms disappear by the boundary condition (\ref{BC}).
Divide both sides by $m'$ and substitute the relation ${u\over m'}={c\over k(m)}$, which is from the equation (\ref{TravelingEqn1}), and obtain
\begin{equation}\label{modified}
\Big(\gamma_mu+{c^2\over k(m)}\Big)+(\gamma_u u+\gamma){u'\over m'}=0.
\end{equation}
Therefore, the system (\ref{modified}) and (\ref{TravelingEqn1}) is equivalent to (\ref{TravelingEqn}) and (\ref{TravelingEqn1}) under the boundary condition $m_+>0$ and $u_+=0$.

Next we consider a relation in the phase plane. (A detailed phase plane analysis is given in Section \ref{sect.phase}.) Since $m>0$ and $m$ is strictly increasing, one may consider $u$ as a function of $m$ by changing the variable. Hence the relation in (\ref{modified}) can be written as
\begin{equation}\label{Exact}
Mdm+Udu=0,
\end{equation}
where $M:=\gamma_mu+{c^2\over k(m)}$ and $U:=\gamma_u u+\gamma$. Then,
$$
M_u=\gamma_{mu}u+\gamma_m=U_m,
$$
which shows the equation for the traveling wave solution (\ref{Exact}) is exact. This exactness is the key advantage of the chemotaxis model and makes a complete picture of traveling waves available. Now the solution is implicitly given by
$$
\psi(m,u):=\gamma u + c^2H(m)={\rm constant},\quad H'(m)={1\over k(m)}.
$$
Since $u\to u_+=0$, $s\to\lim_{m\to m_+} {m\over u}=\infty$, and hence $\gamma\to \ell>0$ as $m\to m_+$, we have the constant is $c^2H(m_+)$. Hence, we have $H=-\int_m^{m_+}{1\over k(\eta)}d\eta$ and obtain
\begin{equation}\label{psi}
\psi(m,u):=\gamma u-c^2\int_m^{m_+}{1\over k(\eta)}d\eta=0,\quad m_-<m<m_+.
\end{equation}
Then, since
$$
\psi_u=\gamma_u\,u+\gamma>0,\quad m_-<m<m_+,
$$
the implicit function theorem implies that there is a function $u=u(m)$ on $(m_-,m_+)$. It is clear that $u(m)\to 0\,(=u_+)$ as $m\to m_+$. Hence, for the existence of the  traveling wave solution of  (\ref{TravelingEqn})--(\ref{BC}), we only need to check if $u(m)\to u_-$ as $m\to m_-$. Therefore, showing a necessary condition for the boundary condition also gives a sufficient condition, which is one of main arguments in the following proofs.

Having the chemotaxis model in the form of (\ref{Eqn2}) is of great advantage since it gives an exact form for the traveling wave solution. This makes all the troublesome analyses simple and a complete picture accessible. Now the proof of Theorems are easy and clear.

\begin{proof}[of Theorem \ref{thm1}]

$(i)$. From the equation (\ref{TravelingEqn1}) we obtain
$$
u=c{m'\over k(m)}.
$$
Since $m>0,u>0$ and $m'>0$ for all $\xi\in\bfR$, the wave speed is strictly positive by the relation, i.e., $c>0$. Suppose that $m_->0$. Then, taking $m\to m_-$ limit for (\ref{psi}) gives
$$
\lim_{m\to m_-}\gamma\Big({m\over u(m)}\Big)u(m)=c^2\int_{m_-}^{m_+}{1\over k(\eta)}d\eta>0.
$$
If $u_-=0$, then $\ds\lim_{m\to m_-}\gamma\big({m\over u(m)}\big)u(m)=\ell\times0=0$ which contradicts the above relation. Hence $u_->0$, which also contradicts to Lemma \ref{lem2}. Therefore $m_-=0$ is a necessary condition.

$(ii)$. Suppose that the motility function is bounded by $\gamma<A$ for some constant $A>0$. Then, the relation (\ref{psi}) is written by
$$
Au\ge c^2\int_{m}^{m_+}{1\over k(\eta)}d\eta>0.
$$
Taking $m\to m_-(=0)$ gives $u_-\ge{c^2\over A}\int_{0}^{m_+}{1\over k(\eta)}d\eta>0$. If $\int_{0}^{m_+}{1\over k(\eta)}d\eta<\infty$, then it contradicts to Lemma \ref{lem2} and hence there is no traveling wave solution. If $\int^{m_+}_{0}{1\over k(\eta)}d\eta=\infty$, then $u(m)\to\infty$ as $m\to0$ and hence there is no traveling wave solution, neither. Therefore, for any boundary condition $u_-\in\bfR$, we conclude that there is no traveling wave solution if $\gamma$ is bounded.
$\hfill\qed$
\end{proof}

\begin{proof}[of Theorem \ref{thm2}]

($\Rightarrow$) We have $u_-=0$ by Lemma \ref{lem2} and $m_-=0$ by Theorem \ref{thm1}$(i)$. If $\gamma$ is bounded, then $u_->0$ by the relation (\ref{psi}). Therefore, $\gamma$ should be unbounded. Since $\gamma$ is smooth on $(0,\infty)$ and $\gamma(s)\to\ell$ as $s\to\infty$, we have $\gamma(s)\to\infty$ as $s\to0$.

($\Leftarrow$) It is enough to show that $u(m)\to u_-\,(=0)$ as $m\to m_-\,(=0)$, where $u=u(m)$ satisfies (\ref{psi}), i.e.,
$$
\gamma\Big({m\over u(m)}\Big) u(m)=c^2\int_m^{m_+}{1\over k(\eta)}d\eta,\quad m_-<m<m_+.
$$
Suppose that $u(m)\to u_0>0$ as $m\to 0$. Then, for any $c$ finite, as $m\to 0$,
\begin{eqnarray*}
&&\gamma\Big({m\over u(m)}\Big) u(m)\to\infty,\\
&&c^2\int_m^{m_+}{1\over k(\eta)}d\eta\to c^2\int_0^{m_+}{1\over k(\eta)}d\eta<\infty,
\end{eqnarray*}
which contradict to each other. Hence $u(m)\to 0$ as $m\to 0$.

The total population of the traveling wave band is computed by (\ref{TravelingEqn1}). Then,
$$
\int_{-\infty}^\infty u(\xi)d\xi
=\int_{-\infty}^\infty{cm'\over k(m)}d\xi
=\int_{0}^{m_+} {c\over k(m)}dm,
$$
which gives (\ref{population}).
$\hfill\qed$
\end{proof}

\begin{proof}[of Theorem \ref{thm3}]
We already have in Theorem \ref{thm1}$(i)$ that $m_-=0$ is a necessary condition.

$(i)$ Differentiate (\ref{psi}) with respect to $m$ and multiply $k(m)$ to obtain
\begin{equation}\label{000}
\gamma'\Big({m\over u(m)}\Big)k(m)+c^2=0, \quad 0<m<m_+.
\end{equation}
Therefore, (\ref{speed}) implies that $u(m)\to u_0$ as $m\to0$ if and only if there exists a traveling wave solution.

$(ii)$ Suppose that $u_->0$. Then, by the relation (\ref{000}), we have
$$
-\lim_{m\to0}\gamma'\Big({m\over u_-}\Big)k(m)=c^2.
$$
which contradicts the assumption (\ref{speed2}). Hence $u(m)\to0$ as $m\to0$ if and only if there exists a traveling wave solution. The total population of the traveling wave band is computed by (\ref{TravelingEqn1}). Then, since $c\ne0$,
$$
\int_{-\infty}^\infty u(\xi)d\xi
=\int_{-\infty}^\infty{cm'\over k(m)}d\xi
=\int_{0}^{m_+} {c\over k(m)}dm=\infty.
$$

$(iii)$ One may similarly obtain $u(m)\to\infty$ as $m\to0$. Therefore, there is no traveling wave solution with a finite boundary condition $u_-<0$.
$\hfill\qed$
\end{proof}

\section{Phase Plane Analysis}\label{sect.phase}

The traveling wave solutions satisfy the implicit relation (\ref{psi}). We study the structure of the traveling wave solution in the phase plane. Note that (\ref{Exact}) gives the regularity of the solution, i.e.,
\begin{equation} \label{regularity}
\frac{du}{dm} =\frac{-1}{\gamma_{u}u+\gamma}\Big(\gamma_{m}u+{c^{2}\over k(m)}\Big),\quad m_-<m<m_+.
\end{equation}
Hence, the derivative of $u=u(m)$ is continuous and bounded. This relation also implies that the monotonicity of the solution is divided along a curve
$$
\gamma_{m}u+{c^{2}\over k(m)}=0.
$$
We will call this curve a separator since the monotonicity of trajectories are different in the two regions separated by this curve. Since $\gamma_mu=\gamma'(s)$, this curve is given by the relation
\begin{equation} \label{separator}
-\gamma'(s)={c^{2}\over k(m)},
\end{equation}
which is the relation used to characterize the traveling waves in Theorem \ref{thm3} as in (\ref{speed}). In this section we analyze the trajectories of traveling waves in the phase plane in terms of the separator.

\begin{figure}
{\large
\psfrag{a}{\hskip -7mm${c^2\over k(m_0)}$} \psfrag{b}{\hskip -7mm${c^2\over k(m_1)}$} \psfrag{c}{$s$}
\psfrag{d}{\hskip -3mm$m_0$}\psfrag{f}{$m_1$} \psfrag{e}{\textcolor{red}{\hskip -2mm$m_+$}} \psfrag{i}{${du\over dm}>0$} \psfrag{j}{${du\over dm}<0$} \psfrag{g}{$m$}\psfrag{h}{$u$}
\hskip 0.9cm
\includegraphics[width=0.9\textwidth]{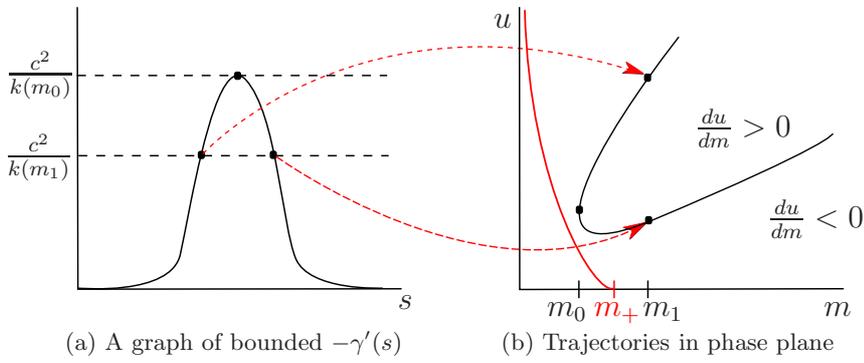}
}

\hskip 1cm(a) A graph of bounded $-\gamma'(s)$~~~~~~~~~~~
(b) Trajectories in phase plane
\caption{\label{fig1}The monotonicity of a traveling wave trajectory in the phase plane is divided by a curve called `separator' in (b). This diagram shows its relation to the graph of $\gamma'$ in (a) when $\gamma'$ is bounded.}
\end{figure}
Consider a case that $-\gamma'(s)$ is bounded (see Figure \ref{fig1}(a)\,). First we find the curve of separator that satisfies (\ref{separator}). Let $-\gamma'(s)\le\kappa_0$ and $-\gamma(s_0)=\kappa_0$ with $s_0\ne0$. Then, since $k(m)\to 0$ monotone as $m\to0$, there exists $m_0>0$ such that ${c^{2}\over k(m_0)}=\kappa_0$. Therefore, a point $(m_0,u_0)$ is on the curve if $s_0={m_0\over u_0}$ or $u_0={m_0\over s_0}$. Since ${c^{2}\over k(m)}>\kappa_0$ for all $m<m_0$, the separator does not enter the region $m<m_0$ (see Figure \ref{fig1}(b)\,). Let $m_1>m_0$. Then, there exists at least one $s>0$ such that ${c^{2}\over k(m_1)}=-\gamma'(s)$. For example, in Figure \ref{fig1}(a), there are two such points. Therefore, there should be two points in the separator corresponding the value $m=m_1$. Since $u={m\over s}$, the point with smaller satisfaction measure $s$ corresponds to the upper point as in Figure \ref{fig1}. A rough sketch of the separator is given in Figure \ref{fig1}(b). If the maximum point is $s_0=0$, then the vertical line $m=m_0$ becomes a asymptote of the separator.

Next, consider a trajectory $u=u(m)$ of a traveling wave in the phase plane that connects an unstable steady state $(m_+,0)$. Since ${du\over dm}<0$ for all $m<m_0$, there is no chance that the trajectory connects the origin as $m\to0$. One might guess that $u(m)\to u_0>0$ as $m\to0$. However, that is not possible by Theorem \ref{thm3}. Clearly,
$$
-\lim_{m\to0}\gamma'\Big({m\over u_0}\Big)k(m)=0
$$
for any $u_0>0$ since $\gamma'$ is bounded and $k(m)\to0$ as $m\to0$. Therefore, Theorem \ref{thm3}$(iii)$ implies that $u(m)\to\infty$ as $m\to0$. Otherwise, there should exist a traveling wave. Therefore, the only possibility is that the trajectory blows up as in Figure \ref{fig1}(b).

The second case is when $\gamma'(s)$ is unbounded. Since $\gamma'(s)\to0$ as $s\to\infty$ and $\gamma$ is smooth on $(0,\infty)$, $-\gamma'(s)\to\infty$ as $s\to0$ (see Figure \ref{fig2}(a)\,). Since $\gamma'$ is unbounded, there exists at least one point in the separator for each $m>0$. We assume that $\gamma'$ is monotone for $s$ small enough as in Theorem \ref{thm3}. Then, there are three possible cases for the limits of the separator as $m\to0$ which are given in Figures \ref{fig2}(b), \ref{fig2}(c) and \ref{fig2}(d). These three cases correspond to the three cases of Theorem \ref{thm3} in that order. In fact the conditions of (\ref{speed}) and (\ref{speed2}) are the ones of (\ref{separator}) for $m$ small.
\begin{figure}
\begin{minipage}[t]{0.5\textwidth}
{\large\psfrag{s}{${s}$}
\psfrag{a}{$~$}
\psfrag{b}{\hskip -6mm$-\gamma'$}
\hskip 4mm
\includegraphics[width=0.7\textwidth]{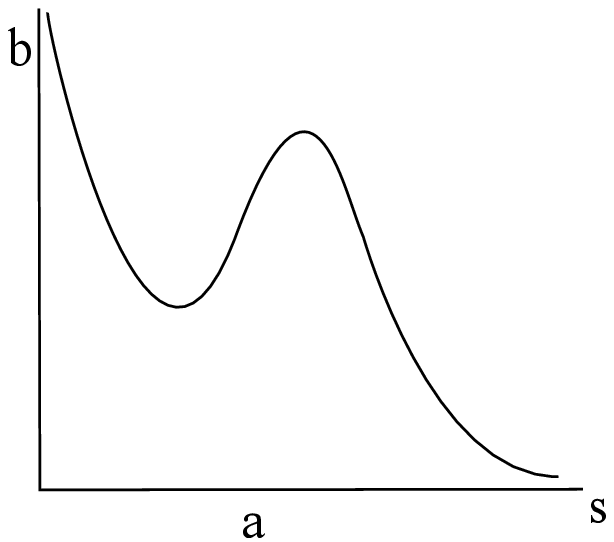}}
\newline(a) A graph of unbounded $-\gamma'(s)$
\end{minipage}
\begin{minipage}[t]{0.5\textwidth}
{\large
\psfrag{i}{${du\over dm}>0$}
\psfrag{j}{${du\over dm}<0$}
\psfrag{m}{$m$}
\psfrag{u}{\hskip -1mm$u$}
\psfrag{x}{\hskip -3mm$m_+$}
\psfrag{y}{\hskip -3mm$m_+$}
\psfrag{z}{\hskip -3mm$m_+$}
\hskip 4mm
\includegraphics[width=0.7\textwidth]{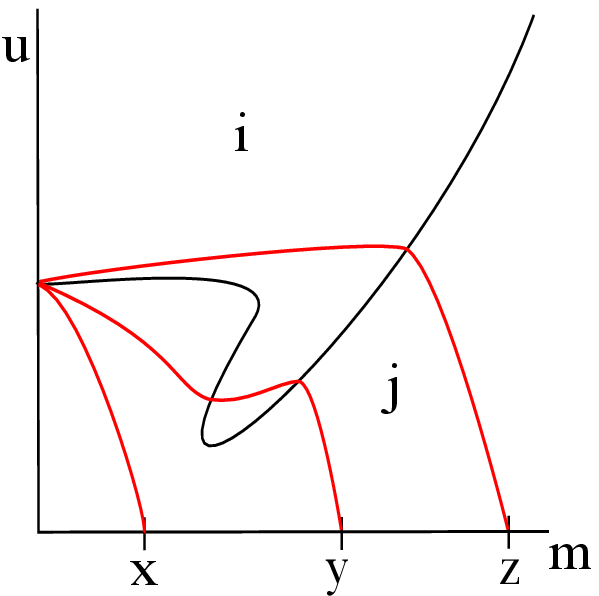}}
\newline(b) Trojectories of case in Thm \ref{thm3}$(i)$
\end{minipage}
\begin{minipage}[t]{0.5\textwidth}
{\large
\psfrag{i}{${du\over dm}>0$}
\psfrag{j}{${du\over dm}<0$}
\psfrag{m}{$m$}
\psfrag{u}{\hskip -1mm$u$}
\psfrag{v}{\hskip -3mm$m_+$}
\psfrag{w}{\hskip -3mm$m_+$}
\hskip 4mm
\includegraphics[width=0.7\textwidth]{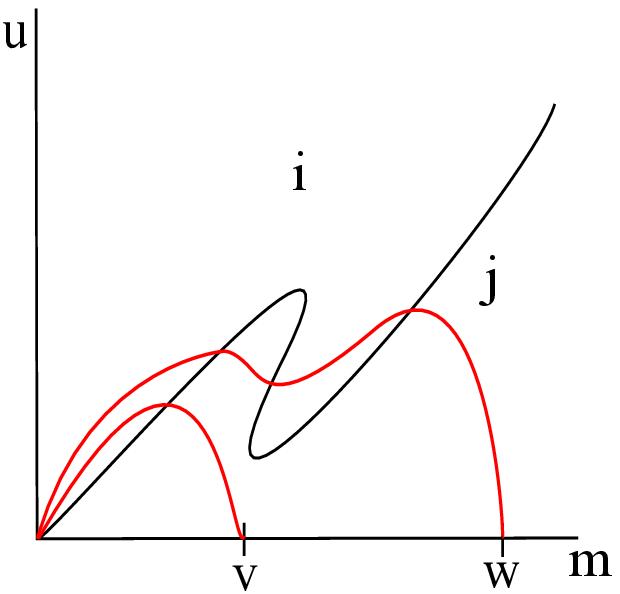}}
\newline(c) Trojectories related to Thm \ref{thm2},\ref{thm3}$(ii)$
\end{minipage}
\begin{minipage}[t]{0.5\textwidth}
{\large
\psfrag{i}{\hskip -3.5mm${du\over dm}>0$}
\psfrag{j}{${du\over dm}<0$}
\psfrag{m}{$m$}
\psfrag{u}{\hskip -1mm$u$}
\psfrag{a}{\hskip -3mm$m_+$}
\hskip 4mm
\includegraphics[width=0.7\textwidth]{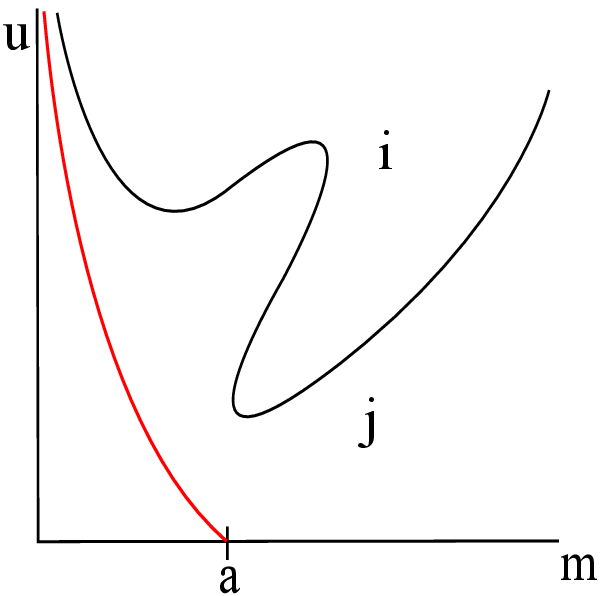}}
\newline(d) Trojectories related to Thm \ref{thm3}$(iii)$
\end{minipage}
\caption{\label{fig2} A graph of $-\gamma'$ is given in (a) when $-\gamma'(s)\to\infty$ as $s\to0$. Three types trajectories of traveling waves are given with dashed lines in (b),(c) and (d). The monotonicity of a traveling wave trajectory is divided by a curve called `separator' given in solid lines.}
\end{figure}

Next, consider a trajectory $u=u(m)$ of a traveling wave in the phase plane that connects an unstable steady state $(m_+,0)$. Theorem \ref{thm3} implies that the trajectory connects the same limit of the separator as $m\to0$ (see Figures \ref{fig2}(b),\ref{fig2}(c) and \ref{fig2}(d)). One can formally see the reason directly from (\ref{psi}). If $u(m)\to u_0>0$ as $m\to0$, the relation gives
$$
\gamma\Big({m\over u_0}\Big) u_0-\int_m^{m_+}{c^2\over k(\eta)}d\eta\to0\mbox{~~~as~~}m\to0.
$$
Therefore, this relation is possible only if two terms diverges with the same speed. If the integral term grows too slow in compare with a given motility function or vice-versa as $m\to0$, then there is no such $u_0>0$. Hence the derivative of both terms with respect to $m$ decides the limit, which are
$$
{\partial \over \partial m}\Big(\gamma\Big({m\over u_0}\Big)u_0\Big) =\gamma'\Big({m\over u_0}\Big),\quad
{\partial \over \partial m}\int_m^{m_+}{c^2\over k(\eta)}d\eta=-{c^2\over k(m)}.
$$
In other words, even if the separator and the traveling wave trajectories satisfies different relation, their behaviors near $m=0$ are decided by the same dynamics and hence they have the same limit as $m\to0$.

\section{Numerical Simulations}

The purpose of this section is to simulate the appearance of a localized traveling wave and compare it to the theoretical one developed in the previous section. Consider a Cauchy problem of our chemotaxis model
\begin{equation}\label{simulation}
\left\{
\begin{aligned}
&u_t= (\gamma(s)u)_{xx},&m_t= -k(m)u,\\
&u(x,0)=0.1\chi_{(-0.2,0.2)},&m(x,0)=0.1,\ \\
\end{aligned}
\right.
\quad x\in\bfR,\ t>0,
\end{equation}
where
\begin{equation}\label{Ex1}
\gamma(s)=0.005(s^{-4}+1),\quad k(m)=5\,m^{1/2}.
\end{equation}
Since both the equation and the initial value have symmetry with respect to the origin, the solution is also so. Hence, the initial population distribution splits in two pieces and then propagates into opposite directions symmetrically. Hence the population size of each piece is $N=0.02$. In this case we have
$$
\int_0^{m_+}{1\over k(m)}dm={1\over 5}\int_0^{0.1}m^{-1/2}dm= {2\sqrt{0.1}\over 5}<\infty.
$$
Therefore, this is the case corresponding to Theorem \ref{thm2} and the traveling wave speed given by (\ref{population}) is
$$
c=0.02\times {5\over 2\sqrt{0.1}}\cong 0.158.
$$

\begin{figure}
\begin{minipage}[t]{0.5\textwidth}
\centering
\includegraphics[width=0.9\textwidth]{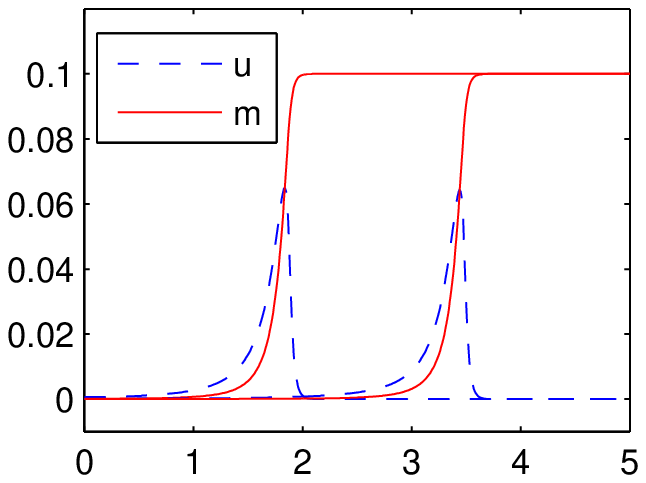}
\newline(a) Traveling waves at two different moments $t=10$ and $20$
\end{minipage}
\begin{minipage}[t]{0.5\textwidth}
\centering
\includegraphics[width=0.9\textwidth]{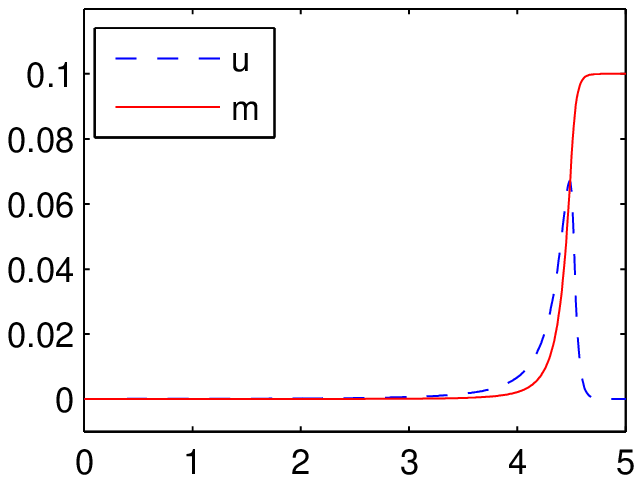}
\newline(b)Pulse type traveling wave for\\ (\ref{Ex1})--(\ref{BCat5}) with $c=0.158$
\end{minipage}
\caption{\label{fig3} [A case with $\int_0^1 {1\over k(m)}dm<\infty$] The numerical simulations for (\ref{simulation}) is given in (a) at two different times $t=10$ and $20$. A pulse type traveling band of finite population is formed and moves with a constant speed. The ode system (\ref{TravelingEqn2}) for traveling wave solution gives the same traveling wave solution.}
\end{figure}

A numerical simulation for the traveling wave moving to right is given in Figure \ref{fig3}(a). The profile of the wave has been displayed at two different incidents $t=10$ and $20$. The numerical wave speed is about $c=0.160$.

Now consider the solution of the ODE (\ref{TravelingEqn})--(\ref{TravelingEqn1}). Since $m_-=u_+=0$ are necessary conditions and $m_+>0$, we assume them. Then, the second order equation in the first row is turned into a first order one after integrating it on $(\xi,\infty)$. Then the equations are written as
\begin{equation}\label{TravelingEqn2}
\left\{
\begin{aligned}
c(\gamma'(s)s-\gamma)u'&=c^2u+\gamma'(s)k(m)u,\\
c\,m'&= k(m)u,
\end{aligned}
\right.
\quad \xi\in\bfR,
\end{equation}
with boundary values
$$
u_-=u_+=m_-=0,\qquad m_+=0.1.
$$
A solution of this traveling wave equation is given in Figure \ref{fig3}(b). For the computation, boundary conditions and the wave speed are given by
\begin{equation}\label{BCat5}
m(5)=0.1,\quad u(5)=10^{-7}, \quad c=0.158.
\end{equation}
The solution is then computed for $x<5$. For the comparison with the numerical simulation of (\ref{simulation}) in Figure \ref{fig3}(a), we have used the same wave speed in this computation. One may observe the same structure of the solution. Since the numerical simulation usually contains extra numerical viscosity, the solution of the ODE system look a little bit steeper.

Next we consider an example of a front type traveling wave corresponding to the case of Theorem \ref{thm3}$(i)$. First, let
\begin{equation}\label{Ex2}
\gamma(s)=0.005(s^{-2}+1),\quad k(m)=10\,m^{3}.
\end{equation}
In this example we took a consumption rate $k$ satisfying $\int_0^{m_+}{1\over k(m)}dm=\infty$.
%
%
%
%
We are looking for a front type traveling wave with $m_+=u_-=0.1$ and $u_+=m_-=0$. Then, the relation (\ref{speed}) is written as
$$
-\lim_{m\to0}\gamma'\Big({m\over 0.1}\Big)k(m)=0.01\Big({m\over 0.1}\Big)^{-3} 10 m^3 =(0.1)^4=c^2.
$$
Hence the traveling wave speed is $c=0.01$. The traveling wave for this case can be numerically computed by solving the ordinary differential equation (\ref{TravelingEqn2}). In Figure \ref{fig4}(a), a traveling wave is given after solving it with boundary conditions,
\begin{equation}\label{BCat0}
m(100)=0.1,\quad u(100)=10^{-7},
\end{equation}
and a wave speed $c=0.01$.
\begin{figure}
\begin{minipage}[t]{0.5\textwidth}
\centering
\includegraphics[width=0.8\textwidth]{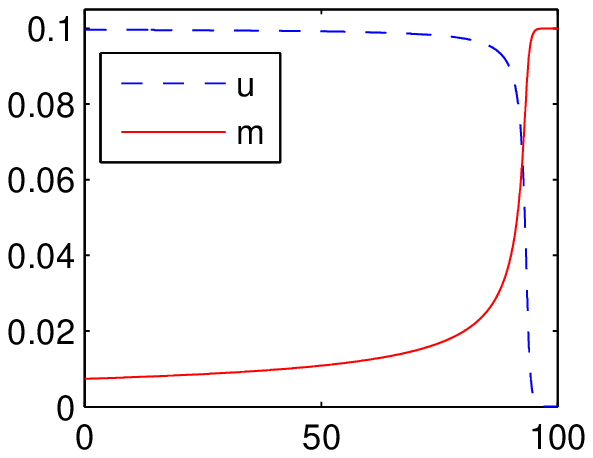}
\newline(a) Front type traveling wave with\\ (\ref{Ex2}), (\ref{BCat0}) and $c=0.01$
\end{minipage}
\begin{minipage}[t]{0.5\textwidth}
\centering
\includegraphics[width=0.8\textwidth]{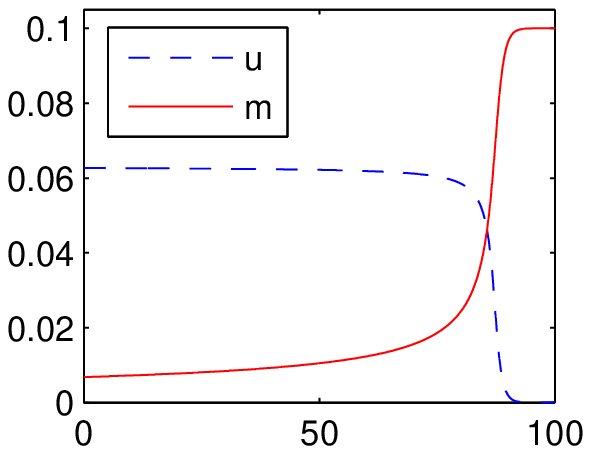}
\newline(b) Front type traveling wave with\\ (\ref{Ex2}), (\ref{BCat0}) and $c=0.005$
\end{minipage}
\begin{minipage}[t]{0.5\textwidth}
\centering
\includegraphics[width=0.8\textwidth]{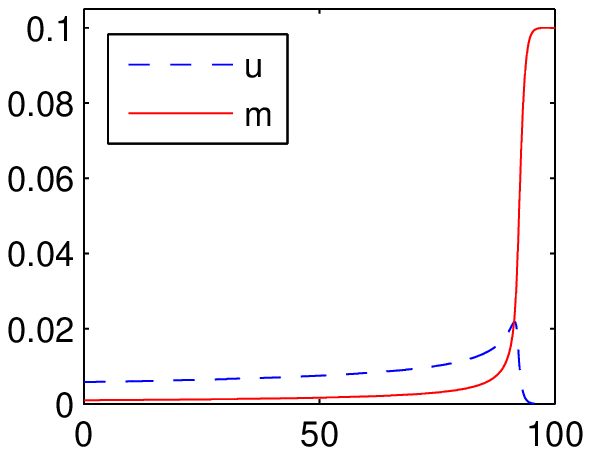}
\newline(c) Pulse type traveling wave with\\ (\ref{Ex3}), (\ref{BCat0}) and $c=0.0077$
\end{minipage}
\begin{minipage}[t]{0.5\textwidth}
\centering
\includegraphics[width=0.8\textwidth]{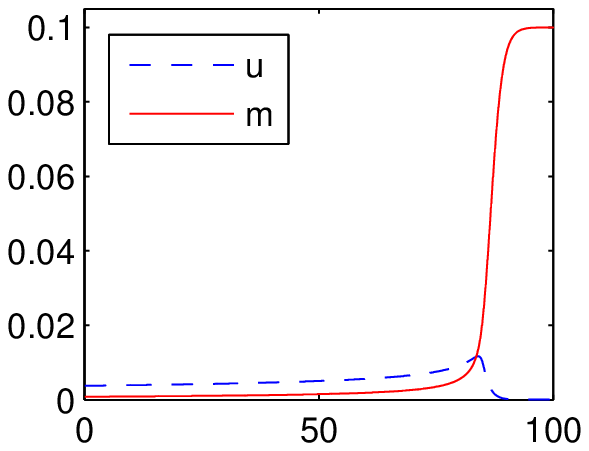}
\newline(d) Pulse type traveling wave with\\ (\ref{Ex3}), (\ref{BCat0}) and $c=0.0039$
\end{minipage}
\begin{minipage}[t]{0.5\textwidth}
\centering
\includegraphics[width=0.8\textwidth]{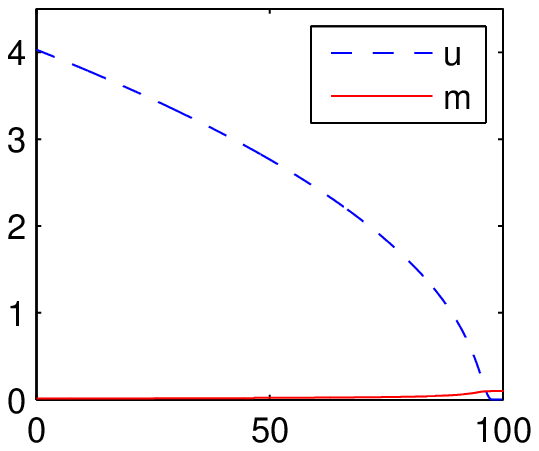}
\newline(e) Blowing up solution with\\ (\ref{Ex4}), (\ref{BCat0}) and $c=0.0224$
\end{minipage}
\begin{minipage}[t]{0.5\textwidth}
\centering
\includegraphics[width=0.8\textwidth]{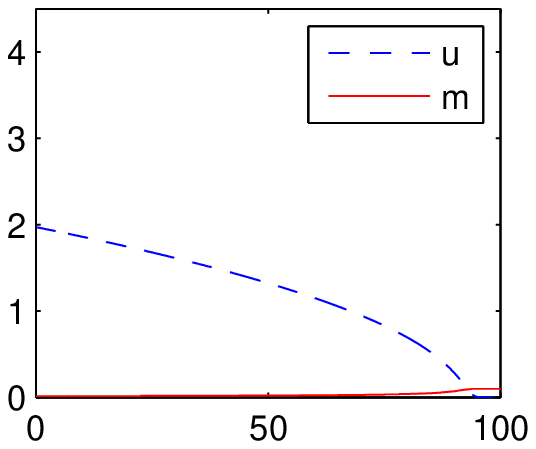}
\newline(f) Blowing up solution with\\ (\ref{Ex4}), (\ref{BCat0}) and $c=0.0112$
\end{minipage}
\caption{\label{fig4} [Three cases with $\int_0^1 {1\over k(m)}dm=\infty$] Front type traveling waves are given in (a) and (b) with different wave speed. The pulse type traveling waves in (c) and (d) are of infinite total population. In (e) and (f) the $u(\xi)\to\infty$ as $\xi\to-\infty$, which shows that there is no traveling wave solution for any boundary value $u_-<\infty$.}
\end{figure}
One may easily observe that the population density $u$ converges to $u_-=0.1$ as $\xi\to-\infty$. In Figure \ref{fig4}(b), the traveling wave with a wave speed $c=0.005$ is given with the same boundary conditions.

Next we consider an example of a localized traveling wave corresponding to the case of Theorem \ref{thm3}$(ii)$. First, let
\begin{equation}\label{Ex3}
\gamma(s)=0.005(s^{-3}+1),\quad k(m)=10\,m^{2}.
\end{equation}
Then, for any given $u_0>0$, one may easily see that
$$
-\lim_{m\to0}\gamma'\Big({m\over u_0}\Big)k(m)=\infty.
$$
Therefore, Theorem \ref{thm3}$(ii)$ implies that there exists a pulse type traveling wave solution for any given wave speed $c>0$ which has infinitely large total population. Therefore, the waves are not distinguishable in terms of the population size or the boundary condition $u_-$. The only measure to distinguish them is the wave speed. In Figures \ref{fig4}(c) and \ref{fig4}(d), two pulse type traveling waves are displayed with wave speeds $c=0.0077$ and $0.0039$, respectively.

Finally we consider an example of blowing up solutions of corresponding to the case of Theorem \ref{thm3}$(iii)$. First, let
\begin{equation}\label{Ex4}
\gamma(s)=0.005(s^{-1}+1),\quad k(m)=10\,m^{4}.
\end{equation}
Then, for any given $u_0>0$, one may easily see that
$$
-\lim_{m\to0}\gamma'\Big({m\over u_0}\Big)k(m)=0.
$$
Therefore, Theorem \ref{thm3}$(iii)$ implies that there is no traveling wave solution for any given boundary data $u_-\ge0$ for any given wave speed $c>0$. In Figures \ref{fig4}(e) and \ref{fig4}(f), two examples of blowing up solution of (\ref{TravelingEqn2}) with wave speeds $c=0.0224$ and $0.0112$, respectively. One may observe that the $u(\xi)$ diverges as $\xi\to-\infty$.

\section{Conclusion}

Chemotaxis is a phenomenon that biological organisms move toward or away from higher concentration of signaling chemicals such as pheromone or nutrients. However, it is puzzling how a single micro-scale organism such as bacteria recognize the macro-scale gradient. There are several models based on sensing spatial or temporal gradients of concentration. However, the main contribution of the classical Keller-Segel model on this issue is that, even if the individuals sense the gradient inaccurately, the chemotactic phenomenon can be observed as a whole group.

The chemotaxis model studied in this paper is based on a different concept. Individual organisms do not have any information of the nutrient gradient. The only assumption is that, if organisms take enough food, they reduce the motility to stay in the place. If not, they increase their motility to find food. Even though individual movements are in a random manner, conditional eagerness for migration produces advection. Hence we may observe chemotaxis phenomenon in the whole group level, but not in the individual level.

Furthermore, the resulting chemotaxis model includes the original Keller-Segel model as a special case that the motility depends only on the nutrient concentration thanks to the relation between chemosensitivity and diffusivity in (\ref{chi-mu}). Instead of breaking this original link, we kept the relation in terms of starvation driven diffusion. In fact, this relation between two terms gives an exact form for the traveling wave solution and makes all the analysis so simple that a complete scenario of traveling wave solutions from pulse type to front type are classified in terms of the motility function and the consumption rate.

\section*{Appendix: A derivation of a starvation driven diffusion}

In this section we introduce a short derivation of a starvation driven diffusion (see \cite{CK} for detailed discussions). Consider a random walk system with a constant walk length $\Delta x$ and a constant jumping time interval $\Delta t$. Let $0<\gamma(x_i)\le1$ be the probability for a particle to depart a grid point $x_i$ at each jumping time. (For a usual random walk system every particle departs at each jumping time and hence $\gamma=1$.) Each particle moves to one of two adjacent grid points, $x_{i+1}$ or $x_{i-1}$, randomly. Let $U(x_i)$ be the number of particles placed at the grid point $x_i$. Then, the particle density is $u=U/\Delta x$. Hence the net flux that crosses a middle point $x_{i+1/2}:={x_i+x_{i+1}\over2}$ is
\begin{eqnarray*}
J\Big|_{x_{i+1/2}}&=&{\gamma(x_i)|\Delta x|u(x_i)\over 2|\Delta t|}-{\gamma(x_{i+1})|\Delta x|u(x_{i+1})\over 2|\Delta t|}\\
&\cong&-{|\Delta x|^2\over 2\Delta t}(\gamma u)_x\Big|_{x_{i+1/2}}.
\end{eqnarray*}
The corresponding diffusion equation comes from the conservation law
$$
u_t=-\div(J)={|\Delta x|^2\over 2\Delta t}(\gamma u)_{xx}.
$$
After a time rescaling we obtain the desired equation
$$
u_t=(\gamma u)_{xx}.
$$
Notice that this derivation is valid since the probability $\gamma$ depends only on the point $x_i$ that the particle departs. For a more discussion including other cases, see Okubo and Levin \cite[\S 5.4]{MR1895041}.

It is reasonable to believe that biological organisms will increase the departing rate if they are starved. Individuals would be starved if there is no food or if there are too many population. Hence, if $m$ is the amount of food available and $u$ is the population, it is reasonable to assume that $\gamma=\gamma(m,u)$ and
$$
\gamma_m\le0\quad\mbox{and}\quad
\gamma_u\ge0.
$$
The process with these properties is called a starvation driven diffusion. If we set $s={m\over u}$, then $s$ is the amount of food that each individuals may obtain on average. We may consider a case that $\gamma$ is a decreasing function of $s$, i.e., $\gamma(m,u)=\gamma({m\over u})$. Then, clearly,
$$
\gamma_m=\gamma'(s){1\over u}\le0 \quad\mbox{and}\quad
\gamma_u=-\gamma'(s){m\over u^2}\ge0.
$$
This is the case we have considered in this paper.

\bibliographystyle{amsplain}
\bibliography{chemotaxis}

\end{document}